\documentclass[10pt]{article}

\usepackage{a4wide}
\usepackage{latexsym}
\usepackage{amssymb}
\usepackage{amsmath}
\usepackage{amsthm}
\usepackage{xspace}
\usepackage{color}
\usepackage[utf8]{inputenc}
\usepackage{enumerate}
\usepackage{authblk}

\newcommand\product{\ensuremath{\ast}\xspace}
\newcommand\DA{\textbf{DA}\xspace}
\newcommand\LDA{\textbf{LDA}\xspace}
\newcommand\D{\textbf{D}\xspace}
\newcommand\V{\textbf{V}\xspace}
\newcommand\DAD{\DA\product\D}

\newcommand{\mso}{\ensuremath{\textup{MSO}(<)}\xspace}

\newcommand{\fow}{\ensuremath{\textup{FO}(<)}\xspace}
\newcommand{\fown}{\ensuremath{\textup{FO}(<,\suc)}\xspace}

\newcommand\suc{\ensuremath{+1}}

\newcommand{\fodw}{\ensuremath{\textup{FO}^2(<)}\xspace}
\newcommand{\fodwn}{\ensuremath{\textup{FO}^2(<,\suc)}\xspace}

\newcommand{\efgame}{Ehrenfeucht-Fra\"iss\'e\xspace}

\newcommand{\croch}[1]{\left\lfloor #1 \right\rfloor}

\newcommand{\nat}{\ensuremath{\mathbb{N}}\xspace}

\newtheorem{theorem}{Theorem}

\newtheorem{proposition}[theorem]{Proposition}

\newtheorem{lemma}[theorem]{Lemma}
\newtheorem{claim}[theorem]{Claim}
\newtheorem{fact}[theorem]{Fact}

\begin{document}

\title{Decidable Characterization of \fodwn and locality of DA}
\author[1]{Thomas Place}
\author[2]{Luc Segoufin}
\affil[1]{Bordeaux University, Labri}
\affil[2]{INRIA \& ENS ULM, Valda}

\maketitle

\begin{abstract}
Several years ago Thérien and Wilke exhibited a decidable
characterization of the languages of words that are definable
in \mbox{\fodwn}~\cite{fodeux}. Their proof relies on three separate
ingredients. The first one is the characterization of the languages
that are definable in \fodw as those whose syntactic semigroup belongs
to the variety \DA. Then, this result is combined with a \emph{wreath
product argument} showing that being definable in \fodwn corresponds
to having a syntactic semigroup in \DAD. Finally, proving that
membership of a semigroup in \DAD is decidable requires a third
ingredient: the ``locality'' of \DA, a result proved
in~\cite{almeida}. In this note we present a new self-contained and
simple proof that definability in \fodwn is decidable. We obtain the
locality of \DA as a corollary.
\end{abstract}

\section{Introduction}

Regular languages form a robust class of languages characterized by
completely dif and only iferent equivalent formalisms such as automata, finite
semigroups or monadic second-order logic, \mso. In particular, the
connection between \mso definability and recognizability by semigroups
has been used to investigate the expressive power of fragments of
\mso. For this purpose, finding \emph{decidable characterizations} of
such fragments often serves as a yardstick. A \emph{decidable
  characterization} is an algorithm which, given as input a regular
language, decides whether it can be defined in the fragment under
investigation. More than the algorithm itself, the main motivation is
the insight given by its proof. Indeed, in order to prove a decidable
characterization, one needs to consider and understand \emph{all}
properties that can be expressed in the fragment.

Usually a decidable characterization is presented by exhibiting a
variety of semigroups \V such that a language is definable in the
fragment if and only if its syntactic semigroup is in \V. Ideally, membership
of a semigroup in \V is defined as a finite set of equations that
need to be satisfied by all elements of the semigroup. Since the
syntactic semigroup of a language is a finite canonical object that can
effectively be computed from any representation of the language, this
yields decidability. The most striking example, known as
McNaughton-Papert-Sch\"utzenberger's Theorem~\cite{sfo,mnpfo}, is the characterization
of first-order logic equipped with a predicate "$<$" denoting the
linear-order over words, \fow. The result states that a regular language is
definable in \fow if and only if its syntactic semigroup is aperiodic
(i.e. satisfies the identity $s^\omega = s^{\omega+1}$ where $\omega$ is the
size of the syntactic semigroup).

Another successful story is the two-variable fragment of \fow.  Actually two
fragments are of interest: \fodw and \fodwn. \fodw is a restriction of \fow
where only two variables may be used (and reused). \fodwn is then obtained by
adding a predicate ''\suc'' for the successor relation. Note that in full
first-order logic, ''\suc'' can be defined from the order ''$<$.'' However,
this requires more than two variables and therefore \fodwn is strictly more
expressive than \fodw.

In~\cite{fodeux}, Thérien and Wilke proved characterizations for both
\fodw and \fodwn. They show that a language is definable in \fodw 
(resp. \fodwn) if and only if its syntactic semigroup is in the variety \DA
(resp. \DAD). However, the arguments used for proving that these two
characterizations are decidable, are very dif and only iferent. For \fodw, this
is immediate as \DA is defined by an equation: a semigroup belongs to
\DA if it satisfies $(st)^\omega t(st)^\omega = (st)^\omega$\footnote{The
authors of~\cite{fodeux} actually use the identity $(str)^\omega
t(str)^\omega = (str)^\omega$ as the definition of \DA. We use here a
simpler identity that is equivalent to it, see for
instance\cite{DK07}.}.

On the other hand, the variety \DAD is constructed from the varieties
\DA and \D using an agebraic product called the \emph{wreath product}
(''\product''). The advantage of this definition is that Thérien and
Wilke are able to obtain their characterization of \fodwn (with \DAD)
as a consequence of their characterization of \fodw (with \DA) using
a an algebraic argument known as the \emph{wreath product principle}.
The downside is that \DAD is not defined using identities and
decidability of its membership is not immediate. In fact there exist
varieties \V with decidable membership such that membership in \V 
\product \D is undecidable\cite{Auinger10}. The special case of \DAD
is solved using the \emph{locality} of \DA, established
in~\cite{almeida}. It follows from the locality of \DA that
$\DAD=\LDA$ where \LDA is the variety of semigroups $S$ such that for
all idempotents $e$ of $S$, $eSe$ is a semigroup in \DA. From this
definition, identities characterizing \LDA can be derived 
from those of \DA: $(esete)^\omega t(esete)^\omega = (esete)^\omega$
(where $e$ is an idempotent) and the decidability of its membership
follows.

In this paper we present a new proof of the characterization of \fodwn
by taking a dif and only iferent approach. We directly show that a language is
definable in \fodwn if and only if its syntactic semigroup satisfies the identity
$(esete)^\omega s(esete)^\omega = (esete)^\omega$. Our proof remains
simple and relies only on elementary combinatorial arguments. We
essentially show that when the equation holds one can reduce the
problem of constructing an \fodwn formula for the language to
constructing an \fodw formula for another language over a modified 
alphabet.

The paper is organized as follows. We start with the necessary
notations.
 The key
part is Section~\ref{sec-fodwn} where we prove that the identity ensures
definability in \fodwn. In Section~\ref{sec-necessary} we give a
standard game argument showing that the equation is implied by
definability in \fodwn.

\section{Notations}

\noindent
{\bf Words and Languages.} We fix a finite alphabet $A$. We denote by $A^+$ the set of all nonempty finite words and by $A^{*}$ the set of all finite words over $A$. We denote the empty word by $\varepsilon$. If $u,v$ are words, we denote by $u \cdot v$ or by $uv$ the word obtained from the concatenation of $u$ and $v$.

For convenience, we only consider languages that do not contain the empty word. That is, a \emph{language} is a subset of $A^*$. In this paper, we consider regular languages, i.e., languages that can be defined by a \emph{nondeterministic finite automata}~(NFA). In the paper, we work with the algebraic representation of regular languages in terms of monoids.

\medskip
\noindent
{\bf Semigroups and Monoids.} A \emph{semigroup} is a set $S$ equipped with an associative operation $s\cdot t$ (often written $st$). A \emph{monoid} is a semigroup $M$ having a neutral element $1_M$, i.e., such that $s\cdot1_M=1_M\cdot s=s$ for all $s\in M$. 

An element $e$ of a semigroup is \emph{idempotent} if $e^2=e$. Given a \emph{finite} semigroup $S$, it is folklore and easy to see that there is an integer $\omega(S)$ (denoted by $\omega$ when $S$ is understood) such that for all $s$ of $S$, $s^\omega$ is idempotent.

Observe that the set $A^*$ equipped with the concatenation operation is a monoid (the neutral element is the empty word ``$\varepsilon$''). Given a monoid $M$ and a morphism $\alpha: A^* \to M$, we say that a language \emph{$L$ is recognized by $\alpha$} if there exists $F \subseteq M$ such that $L = \alpha^{-1}(F)$. It is well known that a language is regular if and only if it can be recognized by a morphism into a \emph{finite} monoid. Finally, from any NFA recognizing some language $L$, one can compute a canonical morphism $\alpha: A^* \to M$ into a finite monoid recognizing $L$: the \emph{syntactic morphism} of $L$ ($M$ is the transition monoid of the minimal deterministic automaton recognizing it). Additionally, the monoid $M$ is called the \emph{syntactic monoid} of $L$ and the semigroup $S = \alpha(A^+)$ is called the \emph{syntactic semigroup} of $L$.

\medskip
\noindent
{\bf Logic.} As usual a word can be seen as a logical structure whose domain is the sequence of positions in the word. We work with unary predicates $P_a$ for all $a \in A$ denoting positions carrying the letter $a$ and two binary predicates $\suc$ and $<$ denoting the successor relation and the order relation among positions. First-order logic is then defined as usual and we denote by \fodw the two variable restriction of \fow and by \fodwn the two variable restriction of \fown. We shall use the two following classical closure properties of \fodw.

\begin{lemma} \label{lem:union}
	Let $A$ be an alphabet and $K,L \subseteq A^*$ which are definable in \fodw. Then, $K \cup L$ is definable in \fodw.
\end{lemma}

\begin{proof}
	Immediate: we may combine formulas defining $K$ and $L$ using disjunction.
\end{proof}

\begin{lemma} \label{lem:concat}
	Let $A$ be an alphabet and $a \in A$ a letter. Let $K \subseteq (A \setminus \{a\})^*$ and $L \subseteq A^*$ which are definable in \fodw. Then, $KaL$ and $LaK$ are definable in \fodw.	
\end{lemma}

\begin{proof}
	We show that $KaL$ is definable in \fodw (the proof for $LaK$ is symmetrical). By hypothesis we have \fodw formulas $\psi$ and $\Gamma$ which define $K$ and $L$ respectively. Since $K \subseteq (A \setminus \{a\})^*$ by construction, a formula $\varphi$ defining $KaL$ is as follows: 
	\[
	\varphi = \exists x\ P_a(x) \land \psi^{\leq} \land \Gamma^{\geq},
	\]
	where $\psi^{\leq}$ is constructed from $\psi$ by replacing all quantifications $\exists y$ by $\exists y (\forall x \leq y \lnot P_a(x))$ while $\Gamma^{\geq}$ is constructed from $\Gamma$ by replacing all quantifications $\exists y$ by $\exists y (\exists x < y  P_a(x))$. It follows from the definitions that $\varphi$ defines $KaL$.
\end{proof}

\section{Characterization of \fodwn}
\label{sec-fodwn}
In this section we prove the characterization of \fodwn:

\begin{theorem} \label{carac-fodwn}
A regular word language $L$ is definable in \fodwn if and only if its syntactic semigroup $S$ satisfies, for all $s,t,e \in S$ with $e$ idempotent:
\begin{equation} \label{eqdad}
(esete)^{\omega}=(esete)^{\omega}t(esete)^{\omega}
\end{equation}
\end{theorem}

There are two directions to prove. That~\eqref{eqdad} is necessary follows from a classical \efgame argument. We state it in the next proposition whose proof is is postponed to Section~\ref{sec-necessary}.

\begin{proposition} \label{prop-onlyif-fodwn}
If a language $L$ is definable in \fodwn, its syntactic semigroup
satisfies~\eqref{eqdad}.
\end{proposition}

The remainder of this section is devoted to the proof of the other direction. We formalize it with the following proposition.

\begin{proposition} \label{prop-if-fodwn}
	Consider a finite monoid $M$, a morphism $\alpha: A^* \to M$ and $S = \alpha(A^+)$. Moreover, assume that $S$ satisfies~\eqref{eqdad}. Then, any language recognized by $\alpha$ is definable in \fodwn.
\end{proposition}

We fix the morphism $\alpha: A^* \to M$ and $S = \alpha(A^+)$ satisfying~\eqref{eqdad} for the proof. Our argument is based on two steps. We first build another alphabet $B$ and a new morphism $\beta: B^* \to M$. Then, we use our hypothesis on $S$ to prove that any language recognized by $\beta$ can be ``approximated'' with another language definable in \fodw (we make this notion precise below). This suffices to show that the languages recognized by $\alpha$ are definable in \fodwn. 

\medskip

We begin with the definition of the new alphabet $B$. We let $\square$ as some symbol which does not correspond to any element in $M$. Moreover, we write $E(S)$ for the set of idempotents in the semigroup $S$ and fix an arbitrary linear order over it. Consider the new alphabet 
\begin{equation*}
B = \{(e,s,f) \mid e,f \in E(S) \cup \{\square\}, s \in M\},
\end{equation*}
Observe that the morphism $\alpha$ can be generalized as a monoid morphism $\beta : B^* \rightarrow M$. Given $e,f \in E(S)$ and $s \in M$, we let $\beta((e,s,f)) = esf$, $\beta((\square,s,f)) = sf$, $\beta((e,s,\square)) = es$ and $\beta((\square,s,\square)) = s$.

We shall mainly be interested in special words of $B^*$ that we call ``\emph{well-formed}''. A word $u= (e_0,s_0,f_0) \cdots (e_{n},s_n,f_{n}) \in B^*$ is \emph{well-formed} if and only if the three following conditions are satisfied:
\begin{enumerate}
	\item $u$ is non-empty.
	\item $e_0 = f_n = \square$.
	\item For all $i < n-1$, $f_i = e_{i+1} \in E(S)$ (in particular $e_i = f_i \neq \square$).
\end{enumerate}

Given three languages $H,K,L \subseteq B^*$, we say that \emph{$H$ coincides with $K$ over $L$} when $H \cap L = K \cap L$. In particular, when $L$ is the language of all well-formed words, we say that $H$ coincides with $K$ over well-formed words. We may now come back to the proof of Proposition~\ref{prop-if-fodwn}. It is proved as a corollary of the two following lemmas:

\begin{lemma} \label{lem:main1}
	There exists a map $\eta: A^* \to B^*$ which satisfies the two following properties:
	\begin{itemize}
		\item For every $w \in A^*$, $\eta(w)$ is well-formed and $\alpha(w) = \beta(\eta(w))$.
		\item For every language $K \subseteq B^*$ which is \fodw-definable, $\eta^{-1}(K) \subseteq A^*$ is \fodwn-definable. 
	\end{itemize}	
\end{lemma}

\begin{lemma} \label{lem:main2}
	For every $s \in M$, there exists a language $K \subseteq B^*$ which is  \fodw-definable and coincides with $\beta^{-1}(s)$ over well-formed words.
\end{lemma}

Before proving the lemmas, let us use them to finish the proof of Proposition~\ref{prop-if-fodwn}. Let $L \subseteq A^*$ which is recognized by $\alpha$. We have to show that $L$ is \fodwn-definable. By definition, we have $F \subseteq M$ such that $L = \alpha^{-1}(F)$. Consequently,
\[
L = \bigcup_{s \in F} \alpha^{-1}(s)
\]
By Lemma~\ref{lem:union}, it remains to show that $\alpha^{-1}(s)$ is \fodwn-definable for every $s \in M$. By Lemma~\ref{lem:main2}, we get $K \subseteq B^*$ which is \fodw-definable and coincides with $\beta^{-1}(s)$ over well-formed words. One may verify from the first assertion in Lemma~\ref{lem:main1} that $\eta^{-1}(K) = \eta^{-1}(\beta^{-1}(s)) =  \alpha^{-1}(s)$. Moreover, it follows from the second assertion in Lemma~\ref{lem:main1} that $\eta^{-1}(K)$ is \fodwn-definable. Altogether, we get that $\alpha^{-1}(s)$ is \fodwn-definable, concluding the proof.

It remains to prove Lemma~\ref{lem:main1} and Lemma~\ref{lem:main2}. We devote a subsection to each proof.

\subsection{Proof of Lemma~\ref{lem:main1}}

We have to define a map $\eta: A^* \to B^*$ satisfying the two assertions in the lemma. Let us point out beforehand that $\eta$ will \textbf{not} be a morphism. The definition is inspired by~\cite{Str85}.

Consider a word $w \in A^*$. We define $\eta(w)$. If $w$ has length smaller than $|S|$ then $\eta(w) = (\square,\alpha(w),\square)$.

Otherwise, assume that $w=a_1 \cdots a_\ell$ with $\ell > |S|$. Fix $k$ such that $1\leq k \leq \ell- |S|$. It follows from a pigeon-hole principle argument that there exist $k \leq i < j \leq k + |S|$ such that: $\alpha(a_k\cdots a_{i}) = \alpha(a_k \cdots a_{j})$. We then have $\alpha(a_k\cdots a_{i}) = \alpha(a_{k}\cdots a_{i}) (\alpha(a_{i+1} \cdots a_{j}))^{\omega}$. This implies that there is an idempotent $e$ such that $\alpha(a_k\cdots a_{i})=\alpha(a_{k}\cdots a_{i})e$. We set $i_k$ as the smallest such $i \geq k$ and $e_k$ as smallest such idempotent for $i_k$. Doing this for all $k$ yields a set $\{i_1,\dots,i_{\ell-|S|}\}$ of indices together with associated idempotents: $e_1,\dots,e_{\ell-S}$. Observe that it may happen that $i_k = i_{k+1}$. For this reason we rename the set of indices as $\{j_1,\dots,j_h\} = \{i_1,\dots,i_{\ell-|S|}\}$ with associated idempotents $f_1,\dots,f_h$ and such that for all $k$, $j_k < j_{k+1}$. 

We then decompose $w$ as $w=w_1\cdots w_{h+1}$ where: $w_1= a_1 \dots a_{j_1} \in A^+$, for all $k \in \{1,\dots,h\}$, $w_k = a_{j_{k-1}+1}\cdots a_{j_k} \in A^+$ and $w_{h+1} = a_{j_{h}+1} \cdots a_\ell \in A^+$. Observe that by construction, for all $k$,  $w_k$ has length smaller than $|S|$ and
\begin{equation}\label{croch}
\alpha(w) = \alpha(w_1)f_1\alpha(w_1) \cdots f_h \alpha(w_{h+1})
\end{equation}
We define $\eta(w) = b_1 \cdots b_{h+1} \in B^*$ with $b_k=(f_{k-1},\alpha(w_k),f_k)$ (we let $f_{0}= f_{h+1} = \square$). This concludes the definition of $\eta: A^* \to B^*$. Before we show that the two assertions in Lemma~\ref{lem:main1} are satisfied, let us provide some more terminology that we shall need for this proof.

Consider a word $w \in A^*$ and the construction described above. We say that a
position $x$ in $w$ is \emph{distinguished} if it corresponds to the leftmost
position of one of the factors $w_k$ of $w$. To any distinguished position $x$
in $w$, one can associate the corresponding position $\widehat{x}$ in
$\eta(w)$.

The following observation will be crucial in the proof. It essentially states that one can test in \fodwn whether a position $x$ of a word in $A^+$ is distinguished as well as the label of the corresponding position $\hat{x}$ in $\croch w$.

\begin{claim} \label{clm:canonic}
	For any $b \in B$ there exists a formula $\alpha_b(x)$ of \fodwn such that
	for any $w \in A^+$ and any position $x$ of $w$ we have
	
	$w \models \alpha_b(x)$ if and only if $x$ is a distinguished position of $w$ such that
	$\widehat{x}$ has label $b$ in $\croch w$.
\end{claim}
\begin{proof}[Proof sketch] This is because by construction the
	neighborhood of $x$ of size $|S|$ determines whether $x$ is distinguished and
	the label of $\widehat{x}$.  
\end{proof}

We may now prove that the two assertions in Lemma~\ref{lem:main1} are satisfied. Observe that for any $w \in A^*$, $\eta(w)$ is well-formed by construction and by~\eqref{croch}, we have $\beta(\eta(w))=\alpha(w)$. Consequently, the first assertion in Lemma~\ref{lem:main1} is satisfied. We now concentrate on proving the second assertion. 

Consider a language $K \subseteq B^*$ which is \fodw-definable. We have to show that the language $\eta^{-1}(K)$ is \fodwn-definable. By hypothesis, we have a formula $\varphi$ of \fodw defining $K$. We use $\varphi$ to construct $\psi \in \fodwn$ defining $\eta^{-1}(K)$. The construction is based on Claim~\ref{clm:canonic}.

We know from Claim~\ref{clm:canonic} that being a distinguished position is definable in \fodwn. Let $\psi$ be the formula constructed from $\varphi$ by restricting all quantifications to quantifications over distinguished positions and replacing all tests $P_b(x)$ by $\alpha_b(x)$. It is immediate from Claim~\ref{clm:canonic} that $\psi$ defines $\eta^{-1}(K)$.

\subsection{Proof of Lemma~\ref{lem:main2}}

We have to show that for every $s \in M$, $\beta^{-1}(s) \subseteq B^*$ coincides over well-formed words with a language definable in \fodw . The proof requires to consider words in $B^*$ that are slightly more general than well-formed words. They correspond to infixes of well-formed words:

A word $w \in B^*$ is pseudo well-formed if either $w = \varepsilon$ or $w = (e_0,s_0,f_0) \cdots (e_{n},s_n,f_{n}) \in B^+$ where for all $i < n-1$, $f_i = e_{i+1} \in E(S)$. Observe that here is no constraint on $e_0$ and $f_n$, they may be any element in $E(S) \cup \{\square\}$. We call $e_0$ the left guard of $w$ and $f_n$ its right guard (they are undefined if $w = \varepsilon$).

We now present three sets of pseudo well-formed words that we shall use in the proof. Consider two elements $e,f \in E(S) \cup \{\square\}$ and a sub-alphabet $C \subseteq B$. We define three sets of words in $C^*$: $P^C[e]$, $S^C[f]$ and $T^C[e,f]$:
\begin{itemize}
	\item If $e\in E(S)$ then $P^C[e]$ contains the empty word $\epsilon$ and all pseudo-well
          words whose right guard is $e$. If $e=\square$, then $P^C[\square] = \{\varepsilon\}$.
	
	\item If $e\in E(S)$ then $S^C[f]$ contains the empty word $\epsilon$ and all pseudo-well
          words whose left guard is $f$. If $e=\square$ then $S^C[\square] = \{\varepsilon\}$.
	
	\item $T^C[e,f]$ contains all non-empty pseudo well-formed words with
          left guard $e$ and right guard $f$. Additionally, if $e = f \neq
          \square$, then we add the empty word $\varepsilon$ to $T^C[e,f]$.	
\end{itemize}
Observe that by definition, $T^B[\square,\square]$ is the set of all well-formed words in $B^*$.

We may now come back to the proof of Lemma~\ref{lem:main2}. Consider $C \subseteq B$ and $t_1,t_2,s \in M$. We define,
\[
L^C_{s}[t_1,t_2] = \{u \in C^* \mid t_1 \cdot \beta(u) \cdot t_2 =s\}	
\]
Observe that for all $e \in E(S) \cup \{\square\}$, $1_M=\beta(\epsilon)\in P^C(e)$.
Observe also that $L^B_{s}[1_M,1_M] = \beta^{-1}(s)$. We prove Lemma~\ref{lem:main2} as a corollary of the following lemma which we prove by induction.

\begin{lemma} \label{lem:main3}
	Let $C \subseteq B$. Consider $e_1,e_2 \in E(S) \cup \{\square\}$, $t_1 \in \beta(P^C[e_1])$, $t_2 \in \beta(S^C[e_2])$. For every $s \in M$. There exists $K \subseteq C^*$ definable in \fodw which coincides with $L^C_{s}[t_1,t_2]$ over $T^C[e_1,e_2]$.
\end{lemma}

Before we prove Lemma~\ref{lem:main3}, we use it to finish the main argument for Lemma~\ref{lem:main2}. Consider $s \in M$. We apply the lemma in the case when $C =B$, $e_1 = e_2 = \square$, and $t_1 = t_2 = 1_M$. This yields $K \subseteq B^*$ definable in \fodw which coincides with $L^B_{s}[1_M,1_M]$ over $T^B[\square,\square]$. This exactly says that $K \subseteq B^*$ is definable in \fodw and coincides with $\beta^{-1}(s)$ over well-formed words, concluding the proof of Lemma~\ref{lem:main2}.

\medskip

We now concentrate on proving Lemma~\ref{lem:main3}. We fix $C \subseteq B$, $e_1,e_2 \in E(S) \cup \{\square\}$, $t_1 \in \beta(P^C[e_1])$, $t_2 \in \beta(S^C[e_2])$. Finally let $s \in M$. We have to construct the language $K \subseteq C^*$ described in the lemma. The argument is an induction on the three following parameters listed by order of importance:
\begin{enumerate}
\item $|C|$.
\item $|t_1M|$.
\item $|Mt_2|$.
\end{enumerate}
We distinguish two cases based on the following definitions. 
\begin{itemize}
	\item We say that $t_1$ is left saturated when for every $f \in E(S) \cup \{\square\}$ and every $u \in T^C[e_1,f]$, $t_1 \in t_1\beta(u)M$.
	\item We say that $t_2$ is right saturated when for every $f \in E(S) \cup \{\square\}$ and every $u \in T^C[f,e_2]$, $t_2 \in M\beta(u) t_2$.
\end{itemize}
We start with the base which happens when $t_1$ and $t_2$ are respectively left and right saturated. Then, we use induction to handle the case when either $t_1$ is not left saturated or $t_2$ is not right saturated.

\medskip\noindent
{\bf Base case: $t_1$ is left saturated and $t_2$ is right saturated.} We use our hypothesis to prove the following lemma:

\begin{lemma} \label{lem:basecase}
	There exists $r \in M$ such that for $t_1 \beta(w) t_2 = r$ for every $w \in T^C[e_1,e_2]$.
\end{lemma}

Before we prove the lemma, let us use it to conclude the base case. We let $r \in M$ be as defined in Lemma~\ref{lem:basecase}. If $r = s$, we define $K = C^*$ and if $r \neq s$, we define $K = \emptyset$. Clearly, $K$ is \fodw-definable in both cases. Moreover, by definition of $r$ in the lemma, it is immediate that,
\[
L^C_{s}[t_1,t_2] \cap T^C[e_1,e_2] = K \cap T^C[e_1,e_2]
\]
This exactly says that $K$ coincides with $L^C_{s}[t_1,t_2]$ over $T^C[e_1,e_2]$, finishing the proof. It remains to prove Lemma~\ref{lem:basecase}.

\begin{proof}[Proof of Lemma~\ref{lem:basecase}]
	We show that for every $w,w' \in T^C[e_1,e_2]$, we have $t_1 \beta(w) t_2 = t_1 \beta(w') t_2$. This clearly implies the lemma.

	Recall that by definition $t_1 \in \beta(P^C[e_1])$, $t_2 \in \beta(S^C[e_2])$. Hence, there exists $v_1 \in P^C[e_1]$ and $v_2 \in S^C[e_2]$ such that $t_1 = \beta(v_1)$ and $t_2 = \beta(v_2)$. This yields $f,f' \in E(S) \cup \{\square\}$ such that $v_1w' \in T^C[f,e_2]$ and $wv_2 \in T^C[e_1,f']$.
	
	Since $t_1$ and $t_2$ are right and left saturated respectively, it follows that $t_1 \in t_1\beta(wv_2)M = t_1\beta(w)t_2 M$ and $t_2 \in M\beta(v_1w')t_2 = Mt_1 \beta(w') t_2$. This yields $x,y \in M$ such that $t_1 \beta(w') t_2 = t_1 \beta(w) t_2 x$ and $t_1 \beta(w') t_2 = yt_1 \beta(w) t_2$. We now obtain,
	\[ 
	\begin{array}{lll}
		t_1 \beta(w) t_2 & = & yt_1 \beta(w) t_2 x                                                                                   \\
		                 & = & y^\omega t_1 \beta(w) t_2 x^\omega                                                                    \\
		                 & = & y^\omega t_1 \beta(w) t_2 x^{\omega+1} \text{~~~ as~\eqref{eqdad}
			implies $x^{\omega+1}=x^\omega$} \\
		                 & = & t_1 \beta(w) t_2 x                                                                                    \\
		                 & = & t_1 \beta(w') t_2  \text{~~~ by definition of $x$}
	\end{array}
	\]
This concludes the proof.	
\end{proof}

\medskip
\noindent
{\bf Induction step: Either $t_1$ is not left saturated or $t_2$ is not right saturated.} We assume that $t_1$ is not left saturated (the other case is symmetrical). We use induction on the first and second parameters (note that induction on the third parameter is used in the symmetrical case). First, we use our hypothesis to prove the following fact.

\begin{lemma} \label{lem:theletter}
	There exists $c = (e,x,f) \in C$ such that for every $v \in C^*$ satisfying $vc \in T^C[e_1,f]$, $t_1 \not\in t_1\beta(vc) M$.
\end{lemma}

\begin{proof}
	By hypothesis, $t_1$ is not left saturated. Hence, there exists $u \in T^C[e_1,f]$ for some $f \in E(S) \cup \{\square\}$ such that $t_1 \not\in t_1\beta(u)M$. Note that $u$ has to be non-empty (clearly, $t_1 \in t_1M$). Finally, we may choose $u$ of minimal length: $u = u'c$ with $c = (e,x,f) \in C$ and $t_1 \in t_1\beta(u')M$. It remains to show that $c \in C$ satisfies the desired property. Consider $v \in C^*$ such that $vc \in T^C[e_1,f]$, we have to show that $t_1 \not\in t_1\beta(vc) M$. There are two cases depending on whether $e \in E(S)$ or $e = \{\square\}$.
	
	If $e = \square$, then, $u'c = u \in T^C[e_1,f]$ and $vc \in T^C[e_1,f]$ imply that $u' = v = \varepsilon$ and $e_1 = \square$. Hence $vc = u$ and we get by definition of $u$ that $t_1 \not\in t_1\beta(vc) M$. We turn to the case when $e \in E(S)$. We proceed by contradiction: assume that $t_1 \in t_1\beta(vc) M$. This yields $r \in M$ such that $t_1 = t_1\beta(vc)r$.	We have the following fact,
	
	\begin{fact} \label{fct:theidem}
		There exists an idempotent $g \in E(S)$ such that $g\beta(c) = \beta(c)$ and $t_1\beta(u')g = t_1 \beta(u')$.
	\end{fact}

	\begin{proof}
          There are two cases depending on whether $t_1\beta(u') = 1_M$ or
          not. In the former case, we get $1_M =
          t_1\beta(vc)r\beta(u')$.
          Clearly $t_1\beta(vc)r\beta(u') \in S$ since $\beta(c) \in S$
          as $e \in E(S)$. Hence, $1_M \in E(S)$ and it suffices to choose
          $g = 1_M$.
		
		We now assume that $t_1\beta(u') \neq 1_M$. We choose $g = e \in E(S)$ Clearly, $e\beta(c) = \beta(c)$ since $c = (e,x,f)$. Moreover, we have $t_1 \in \beta(P^{C}[e_1])$ by definition and $u'c = u \in T^C[e_1,f]$. This yields, $t_1 \beta(u') \in \beta(P^{C}[e])$ and since $t_1\beta(u') \neq 1_M$ this implies $t_1\beta(u')e = t_1 \beta(u')$.		
	\end{proof}
	
	We may now finish the proof. Recall that $t_1 \in t_1\beta(u')M$ by hypothesis which yields $r' \in M$ such that $t_1 = t_1\beta(u')r'$. Since we also have $t_1 = t_1\beta(vc)r$, this yields the following,
	\[
	\begin{array}{lll}
	t_1\beta(u') & = & t_1\beta(u')r'\beta(vc)r\beta(u')                                                                      \\
	& = & t_1\beta(u')gr'\beta(v)g\beta(c)r\beta(u')g \quad
	\text{Using Fact~\ref{fct:theidem}}                         \\
	& = & t_1\beta(u')(gr'\beta(v)g\beta(c)r\beta(u')g)^{\omega}                                                 \\
	& = & t_1\beta(u')\beta(c)r\beta(u')(gr'\beta(v)g\beta(c)r\beta(u')g)^\omega
	\quad\text{Using~\eqref{eqdad}}
	\end{array}
	\]	
	\noindent
	Consequently, we get $y \in M$ such that $t_1\beta(u') = t_1\beta(u'c) y = t_1\beta(u) y$. Since $t_1 = t_1\beta(u')r'$, we then obtain $t_1 = t_1\beta(u) yr'$. Hence, $t_1 \in t_1\beta(u) M$ which contradicts the definition of $u$.
\end{proof}

We may now finish the proof. We first use induction to build several \fodw-definable languages. We then combine them into another \fodw-definable language $K$ that coincides with $L^C_{s}[t_1,t_2]$ over $T^C[e_1,e_2]$ as desired. 

Let $D = C \setminus \{c\}$. We first handle the words in $D^*$: we build a language $H \subseteq D^*$ which coincides with $L^C_{s}[t_1,t_2]$ over $T^D[e_1,e_2]$. For every $r \in M$, induction on our first parameter (the size of $C$) yields a language $H_r \subseteq D^*$ definable in \fodw which coincides with $L^D_{r}[1_M,1_M]$ over $T^D[e_1,e_2]$. We define,
\[
H = \bigcup_{\{r \in M \mid t_1rt_2 = s\}} H_r
\] 
Clearly, $H \subseteq D^*$ is definable in \fodw by Lemma~\ref{lem:union}. Moreover, one may verify the following fact from the definition.

\begin{fact} \label{fct:subalph}
	$H$ coincides with $L^C_{s}[t_1,t_2]$ over $T^D[e_1,e_2]$.
\end{fact}

We now take care of the words in $C^* \setminus D^*$ (i.e. the ones that
contain at least one letter ``$c$''). Recall that $c = (e,x,f)$.

Let $R \subseteq M$ be as follows:
\[
R = \{\beta(v) \mid \text{$v \in C^*$ and $vc \in T^C[e_1,f]$}\}
\]

For every $r \in R$, induction on our first parameter (the size of $C$), yields
a language $U_r \subseteq D^*$ definable in \fodw which coincides with
$L^D_{r}[1_M,1_M]$ over $T^D[e_1,e]$.

Moreover, by Lemma~\ref{lem:theletter}, we know that for every $r \in R$,
$t_1 \not\in t_1r\beta(c) M$. Clearly, this yields that
$|t_1r\beta(c) M| < |t_1M|$. Hence, induction on our second parameter (the size
of $t_1M$) yields a language $V_r \subseteq C^*$ definable in \fodw which
coincides with $L^C_{s}[t_1r\beta(c),t_2]$ over $T^C[f,e_2]$.

We are now ready to define the language $K \subseteq C^*$ described in Lemma~\ref{lem:main3}. We let,
\[
K = H \cup \bigcup_{r \in R} U_r c V_r
\]
Let us first explain why $K$ is definable in \fodw. By Lemma~\ref{lem:union}, it suffices to show that every language in the union is definable in \fodw. We already know this for $H$. Moreover, given $r \in R$, $U_r,V_r$ are definable in \fodw by definition and $U_r \subseteq D^*$ with $c \not \in D$. Hence, Lemma~\ref{lem:concat} yields that $U_r c V_r$ is definable in \fodw. Altogether, we get that $K$ is definable in \fodw.

\medskip

It remains to verify that $K$ coincides with $L^C_{s}[t_1,t_2]$ over $T^C[e_1,e_2]$. Hence, we fix $w \in T^C[e_1,e_2]$ and show that $w \in K$ if and only if $w \in L^C_{s}[t_1,t_2]$.

Assume first that $w \in K$. We show that $w \in L^C_{s}[t_1,t_2]$. If $w \in H \subseteq D^*$, this is immediate by Fact~\ref{fct:subalph}. Otherwise, $w \in U_r c V_r$ for some $r \in R$. Hence, $w = w_1 c w_2$ with $w_1 \in U_r \subseteq D^*$ and $w_2 \in V_r$. Since $c = (e,x,f)$ and $w \in T^C[e_1,e_2]$, it is immediate that $w_1 \in T^D[e_1,e]$ and $w_2 \in T^C[f,e_2]$. Therefore, by definition of $U_r$ and $V_r$, we get that $w_1 \in L^D_{r}[1_M,1_M]$ (i.e. $\beta(w_1) = r$) and $w_2 \in L^C_{s}[t_1r\beta(c),t_2]$ (i.e. $t_1r\beta(c)\beta(w_2)t_2 = s$). Altogether, this yields,
\[
t_1\beta(w)t_2 = t_1\beta(w_1)\beta(c)\beta(w_2)t_2 = t_1r\beta(c)\beta(w_2)t_2 = s
\]
Hence, $w \in L^C_{s}[t_1,t_2]$ by definition.

Assume now that $w \in L^C_{s}[t_1,t_2]$. We show that $w \in K$. If $w \in D^*$, it is immediate from Fact~\ref{fct:subalph} that $w \in H \subseteq K$. Otherwise, $w$ contains the letter $c$: $w =  w_1 c w_2$ with $w_1 \in D^*$ and $w_2 \in C^*$ (i.e. the highlighted $c$ is the leftmost one in $w$). Since $w \in  T^C[e_1,e_2]$ and $c = (e,x,f)$, this yields $w_1 \in T^D[e_1,e]$, $w_1c \in T^C[e_1,f]$ and $w_2 \in T^C[f,e_2]$. In particular, $w_1c \in T^C[e_1,f]$ means that $\beta(w_1) = r \in R$ by definition. Hence, $w_1 \in L^D_{r}[1_M,1_M]$ which yields $w_1 \in U_r$ by definition of $U_r$ since $w_1 \in T^D[e_1,e]$. Moreover, since $w \in L^C_{s}[t_1,t_2]$, we have $t_1\beta(w)t_2 = s$ which yields $t_1r\beta(c) \beta(w_2) t_2 = s$. Consequently $w_2 \in  L^C_{s}[t_1r\beta(c),t_2]$ which yields $w_2 \in V_r$ by definition of $V_r$ since $w_2 \in T^C[f,e_2]$. Altogether, we obtain $w = w_1c w_2 \in H_r c V_r \subseteq K$ which concludes the proof.

\section{Proof of necessity of~\eqref{eqdad}}\label{sec-necessary}

The proof of Proposition~\ref{prop-onlyif-fodwn} is a simple classical
\efgame argument. We include a sketch below for completeness. We begin
with the definition of the \efgame game associated to \fodwn.

There are two players, Duplicator and Spoiler and the board consists
in two words and a number $k$ of rounds that is fixed in advance. At
any time during the game there is one pebble placed on a position of
one word and one pebble placed on a position of the other word and
both positions have the same label. If the initial position is not
specified, the game starts with the two pebbles placed on the first
position of each word. Each round starts with Spoiler moving one of
the pebbles inside its word from its original position $x$ to a new
position $y$. Duplicator must answer by moving the pebble in the
other word from its original position $x'$ to a new position
$y'$. Moreover, the positions $x'$ and $y'$ must satisfy the same
atomic formulas as $x$ and $y$, i.e. the same predicates among $<$,
$+1$ and the label predicates.

If at some point Duplicator cannot answer Spoiler's move, then Spoiler 
wins the game. If Duplicator is able to respond to all $k$ moves of
Spoiler then she wins the game. Winning strategies are defined as
usual. If Duplicator has a winning strategy for the $k$-round
game played on the words $w,w'$ then we say that $w$ and $w'$ are
$k$-equivalent and denote this by $w \simeq_{k}^{+} w'$. The
following result is classical and simple to prove. 

\begin{lemma}[Folklore]\label{lemma-ef-game}
If $L$ is definable in \fodwn then there is a $k$ such that $w
\simeq_{k}^{+} w'$ implies $w\in L$ if and only if $w'\in L$.
\end{lemma}

We can now use Lemma~\ref{lemma-ef-game} to prove
Proposition~\ref{prop-onlyif-fodwn}.

\begin{proof}[Proof of Proposition~\ref{prop-onlyif-fodwn}]
Let $L$ be a language definable in \fodwn. Let $\alpha: A^*
\rightarrow M$ its syntactic morphism and $S = \alpha(A^+)$ its syntactic semigroup. Let $s$, $t$ and $e$ be 
elements of $S$ with $e$ idempotent. Let $U,V,E$ be non-empty words such that
$s=\alpha(U)$, $t=\alpha(V)$, $e=\alpha(E)$. For all $k \in \nat$, let
$w_k$ be the word $(E^kUE^kVE^k)^{k\omega}$ and let $w'_k$ be the word
$(E^kUE^kVE^k)^{k\omega}V(E^kUE^kVE^k)^{k\omega}$. Note that for all
$k$, $\alpha(w_k)$ is $(esete)^\omega$ while $\alpha(w'_k)$ is
$(esete)^\omega t(esete)^\omega$. 

In view of Lemma~\ref{lemma-ef-game}, it is enough for each number $k$
and each words $u_\ell,u_r$, to give a winning strategy for Duplicator
in the $k$-move \efgame game played on $u_\ell w_k u_r$ and $u_\ell
w'_k u_r$.

This is done by induction on the number $i$ of remaining moves.  At
each step of the game one pebble is at position $x$ of $u_\ell w_k
u_r$ and another one is at position $x'$ of $u_\ell w'_k u_r$. The
inductive hypothesis $H(i)$ that Duplicator maintains is:
\begin{enumerate}
\item $x$ and $x'$ have the same label.
\item If $x$ is in a copy of $E$ (resp. $u_\ell$, $u_r$, $U$, $V$)
  then $x'$ is in a copy of $E$ (resp. $u_\ell$, $u_r$, $U$, $V$) at
  the same relative position as $x$.
\item If $x$ has less than $i$ blocks $(E^kUE^kVE^k)$ to its left (resp. to
    its right) then $x'$ is at the same distance as $x$ from the beginning of
    the word (resp. from the end of the word).
\end{enumerate}

It is immediate to check that $H(i)$ holds at the beginning of the
game. It is also simple to verify that this inductive hypothesis can
be maintained during $k$ moves of the game.
\end{proof}

\section{Conclusion}
We have shown that languages definable in \fodwn are exactly those whose
syntactic semigroup satisfies $(esete)^\omega t(esete)^\omega =
(esete)^\omega$. In other words and with abuse of notations we have shown that
$\fodwn=\LDA$.

Recall from~\cite{fodeux} that languages definable in \fodw are exactly those
whose syntactic semigroup is in the variety \DA. From this and a ``wreath
product argument'', essentially Lemma~\ref{lem:main1}, it follows that
languages definable in \fodwn are exactly those whose syntactic semigroup is in
\DAD.

Therefore it follows from our result that $\DAD=\LDA$.
This in turns is equivalent to the locality of \DA (see for example \cite{Til87}).

\bibliographystyle{plain}
\bibliography{main}

\begin{thebibliography}{1}

\bibitem{almeida}
Jorge Almeida.
\newblock A syntactical proof of locality of {{\bf DA}}.
\newblock {\em International Journal of Algebra and Computation},
  6(2):165--177, 1996.

\bibitem{Auinger10}
Karl Auinger.
\newblock On the decidability of membership in the global of a monoid
  pseudovariety.
\newblock {\em {Intl. Journal of Algebra and Computation (IJAC)}},
  20(2):181--188, 2010.

\bibitem{DK07}
Volker Diekert and Manfred Kufleitner.
\newblock {On First-Order Fragments for Words and Mazurkiewicz Traces}.
\newblock In {\em Intl. Conf. on Developments in Language Theory ({DLT})},
  pages 1--19, 2007.

\bibitem{mnpfo}
Robert McNaughton and Seymour Papert.
\newblock {\em Counter-Free Automata}.
\newblock {MIT} Press, 1971.

\bibitem{sfo}
Marcel~Paul Sch{\"u}tzenberger.
\newblock On finite monoids having only trivial subgroups.
\newblock {\em Information and Control}, 8:190--194, 1965.

\bibitem{Str85}
Howard Straubing.
\newblock Finite semigroup varieties of the form {V*D}.
\newblock {\em Journal of Pure and Applied Algebra}, 36:53--94, 1985.

\bibitem{fodeux}
Denis Th{\'e}rien and Thomas Wilke.
\newblock Over words, two variables are as powerful as one quantifier
  alternation.
\newblock In {\em Symp. on the Theory of Computing (STOC)}, pages 234--240,
  1998.

\bibitem{Til87}
Bret Tilson.
\newblock Categories as algebra: An essential ingredient in the theory of
  monoids.
\newblock {\em Journal of Pure and Applied Algebra}, 48(1-2):83--198, 1987.

\end{thebibliography}

\end{document}